\documentclass[a4paper,11pt]{article}
\usepackage[letterpaper, margin=1.2in]{geometry}
\usepackage[utf8]{inputenc}

\usepackage{algorithm}
\usepackage[noend]{algpseudocode}

\usepackage{comment}

\newtheorem{theorem}{Theorem}[section]

\newenvironment{proof}[1][Proof]{\begin{trivlist}
\item[\hskip \labelsep {\bfseries #1}]}{\end{trivlist}}

\newcommand{\qed}{\nobreak \ifvmode \relax \else
      \ifdim\lastskip<1.5em \hskip-\lastskip
      \hskip1.5em plus0em minus0.5em \fi \nobreak
      \vrule height0.75em width0.5em depth0.25em\fi}

\usepackage[backend=bibtex,
style=numeric,
bibencoding=ascii,
style=ieee,
sorting=nty
]{biblatex} 
\title{Efficient Dispersion of Mobile Robots on Graphs}
\author{Ajay D. Kshemkalyani\footnote{Department of Computer Science, University of Illinois at Chicago. Email: ajay@uic.edu}
\and Faizan Ali \footnote{Department of Computer Science, University of Illinois at Chicago. Email: fali28@uic.edu}}
\date{}

\begin{document}

\maketitle

\begin{abstract}
The dispersion problem on graphs requires $k$ robots placed arbitrarily at the $n$ nodes of an anonymous  graph, where $k \leq n$, to coordinate with each other to reach a final configuration in which each robot is at a distinct node of the graph. The dispersion problem is important due to its relationship to graph exploration by mobile robots, scattering on a graph, and load balancing on a graph. In addition, an intrinsic application of dispersion has been shown to be the relocation of self-driven electric cars (robots) to recharge stations (nodes). 
We propose three efficient algorithms to solve dispersion on graphs.
Our algorithms require $O(k \log\Delta)$ bits at each robot, and $O(m)$ steps running time, where  $m$ is the number of edges and $\Delta$ is the degree of the graph. 
The algorithms differ in whether they address the synchronous or the asynchronous system model, and in what, where, and how data structures are maintained. 
\end{abstract}

\section{Introduction}
\label{intro}

\subsection{Background and Motivation}
The problem of dispersion of mobile robots, which requires the robots to spread out evenly in a region, has been explored in the literature \cite{DBLP:conf/wafr/HsiangABFM02}.
The dispersion problem on graphs, formulated by Augustine and Moses Jr.  \cite{DBLP:conf/icdcn/AugustineM18}, requires $k$ robots placed arbitrarily at the $n$ nodes of an anonymous  graph, where $k \leq n$, to coordinate with each other to reach a final configuration in which each robot is at a distinct node of the graph. This problem has various applications; for example, an intrinsic application of dispersion has been shown to be the relocation of self-driven electric cars (robots) to recharge stations (nodes) \cite{DBLP:conf/icdcn/AugustineM18}. Recharging is a time-consuming process and it is better to search for a vacant recharge station than to wait. In general, the problem is applicable whenever we want to minimize the total cost of $k$ agents sharing $n$ resources, located at various places, subject to the constraint that the cost of moving an agent to a different resource is much smaller than the cost of multiple agents sharing a resource.

The dispersion problem is also important due to its relationship to graph exploration by mobile robots, scattering on a graph, and load balancing on a graph. These are fundamental problems that have been well-studied over the years by varying the system model and assumptions. Although some works consider these problems in general graphs, many other works consider specific graphs like grids, trees, and rings. 

\subsection{Our Results}
\label{results}

Our results assume that robots have no visibility and can only communicate with other robots present at the same node as themselves. The robots are deterministic, and are distinguishable. The undirected graph, with $m$ edges, $n$ nodes, diameter $D$, and degree $\Delta$, is anonymous, i.e., nodes have no labels. Nodes also do not have any memory but the ports (leading to incident edges) at a node have locally unique labels. 

\begin{table*}[t]
\caption{Comparison of the proposed algorithms for dispersion on graphs.}
\label{tab:complexity}
\begin{center}
\begin{tabular}{ccccc}\hline
  Algorithm & Model & Memory Requirement & Time & Features \\
  & & at Each Robot & Complexity & \\
  \hline \hline
  Helping-Sync & Sync. & $O(k \log \,\Delta)$ bits & $O(m)$ steps & need to know $m$ \\
  & & & & for termination \\
  \hline
  Helping-Async & Async. & $O(k \log \, \Delta)$ bits & $O(m)$ steps & no termination \\
  \hline
  Independent-Async & Async. & $O(k \log \, \Delta)$ bits & $O(m)$ steps & no termination \\
  \hline
  \hline
\end{tabular}
\end{center}
\end{table*}

We provide three efficient algorithms to solve dispersion in both the synchronous and asynchronous system models. Our algorithms require $O(k \log\Delta)$ bits at each robot, and $O(m)$ steps running time. We assume that the robots do not know any of the graph parameters $n$, $m$, $D$, or $\Delta$ in the algorithms. It is sufficient if $O(k\log\Delta)$ bits are provisioned at each robot. The following is an overview of our algorithms; the upper bound results are given in Table~\ref{tab:complexity}.
\begin{enumerate}
\item
For the synchronous model, we present algorithm {\em Helping-Sync} which needs $O(k\log\Delta)$ bits per robot and $O(m)$ steps time complexity; for this synchronous algorithm, we assume robots know $m$ if termination is to be achieved. In this algorithm, docked robots, defined as robots that have reached their nodes in the final configuration, help visiting robots by maintaining data structures on their behalf.
\item
Algorithm {\em Helping-Async} is the asynchronous version of {\em Helping-Sync} and has the same time complexity $O(m)$ and same space complexity of $O(k\log\Delta)$ bits per robot; however this algorithm requires each docked robot 
to remain active and help other visiting robots.
\item
Algorithm {\em Independent-Async} has the same complexity ($O(m)$ time steps and $O(k\log\Delta)$ bits per robot) and features as Algorithm {\em Helping-Async}; it differs in what, how and where data structures are maintained. Here, each robot maintains its own data structures, as opposed to {\em Helping-Async} where docked robots help visiting robots by maintaining data structures on their behalf.
\end{enumerate}
Although the asynchronous algorithms, technically speaking, do not terminate because the docked robots need to be awake to relay local information to visiting robots, we state their time complexity (to be $O(m)$)  because at most $O(m)$ steps are required for each robot to perform active computations and movements until it docks at a node; after that, a docked robot merely passively helps visiting robots (until they find a node to dock).

\subsection{Related Work}
The dispersion problem on graphs was formulated by Augustine and Moses Jr. \cite{DBLP:conf/icdcn/AugustineM18}. 
They showed a lower bound of $\Omega(D)$ on the time complexity, and an independent lower bound of $\Omega(\log\,n)$ bits per robot, to solve dispersion. They then gave several dispersion algorithms for specific types of graphs, assuming the synchronous computation model. 
Besides giving dispersion algorithms for paths, rings, trees, rooted trees (a rooted tree has all the robots at the same node in the initial configuration), and rooted graphs (a rooted graph has all the robots at the same node in the initial configuration), they gave two algorithms for general graphs in which the robots can be at arbitrary nodes in the initial configuration. The first algorithm uses $O(\log \,n)$ bits at each robot and $O(\Delta^D)$ rounds, whereas the second algorithm uses $O(n \log \,n)$ bits at each robot and $O(m)$ rounds. We claim that, unfortunately, both these algorithms are incorrect. Both algorithms use variants of Depth First Search (DFS), but may fail to search the graph completely, backtrack incorrectly, and can get caught in cycles while backtracking. This also renders their complexity results incorrect. The problems arise because concurrent searches of the graph by different robots interfere with one another in these algorithms. 
Further, the second algorithm executes for an insufficient number of rounds.
Our work considers dispersion in (unrooted) graphs wherein the robots can be at arbitrary nodes in the initial configuration, for both the synchronous and the asynchronous computation models. We consider general graphs rather than restricted graphs like grids, trees, and rings.

The dispersion problem on graphs is closest to the problem of graph exploration by robots. In the graph exploration problem, the objective is to visit all the nodes of the graph. There are many results for this problem. Several works assume specific topologies such as trees \cite{DBLP:journals/talg/AmbuhlGPRZ11,DBLP:conf/sirocco/DisserMNSS17,DBLP:journals/networks/FraigniaudGKP06,DBLP:journals/jco/HigashikawaKLT14}. For general graphs, the results depend on the different system models and assumptions such as the following.
\begin{enumerate}
    \item  what parameters of the graph are known to the robots, 
    \item  whether the graph is anonymous, 
    \item  whether memory is allowed at robots \cite{DBLP:journals/tcs/FraigniaudIPPP05}, 
    \item  whether memory is allowed at the nodes \cite{DBLP:journals/talg/CohenFIKP08}, 
    \item  whether knowledge of the incoming ports through which a robot enters nodes is allowed \cite{DBLP:journals/tcs/FraigniaudIPPP05}, 
    \item  whether exploration is by a single robot or cooperating robots \cite{DBLP:journals/trob/BrassCGX11,DBLP:conf/icarcv/BrassVX14,DBLP:journals/iandc/DereniowskiDKPU15},
    \item  if exploration is by multiple robots, whether robots are allowed to communicate under the local communication model or the global communication model \cite{DBLP:journals/trob/BrassCGX11,DBLP:conf/icarcv/BrassVX14,DBLP:journals/iandc/DereniowskiDKPU15},
    \item  if exploration is by multiple robots, whether robots are colocated or dispersed in the initial configuration,
    \item  whether we are designing a solution that is time optimal, or space optimal,
    \item  whether the bounds on memory are subject to time optimality solutions,
    \item  whether termination of the robot is required (and if so, whether at the starting node) or it is to perpetually traverse the graph
\end{enumerate}
We now review a few of the closest results. 
Fraigniaud et al. \cite{DBLP:journals/tcs/FraigniaudIPPP05} showed that using only memory at a robot, the robot can explore an anonymous graph using $\theta(D\log\Delta)$ bits based on a $(D+1)$-depth restricted DFS. They did not analyze the time complexity, which turns out to be $\sum_{i=1}^{D}O(\Delta^i) = O(\Delta^{D+1})$ which is very high. Their algorithm has no mechanism to avoid getting caught in cycles and the only way out of cycles is the depth-restriction on the DFS. The robot also requires knowledge of $D$ to terminate.
Cohen et al. \cite{DBLP:journals/talg/CohenFIKP08} gave two DFS-based algorithms with $O(1)$ memory at the nodes. The first algorithm uses $O(1)$ memory at the robot and 2 bits memory at each node to traverse the graph. The 2 bits memory at each node is initialized by short labels in a pre-processing phase which takes time $O(mD)$. Thereafter, each traversal of the graph takes up to $20m$ time steps. The second algorithm uses $O(\log\Delta)$ bits at the robot and 1 bit at each node to traverse the graph. The 1 bit memory at each node is initialized by short labels in a pre-processing phase which takes time $O(mD)$. Thereafter, each traversal of the graph takes up to $O(\Delta^{10}m)$ time steps.
Dereniowski et al. \cite{DBLP:journals/iandc/DereniowskiDKPU15} studied the trade-off between graph exploration time and number of robots, assuming that (i) nodes have unique identifiers, (ii) when visiting a node, a list of all its neighbors is also known, (iii) all the robots are located at one node in the initial configuration, (iv) robots have unique identifiers, and (v) there is no bound on the memory of robots, which construct a map of the previously visited subgraph. The authors considered results in both the local communication model, as well as the global communication model. The main contribution is an exploration strategy for a polynomial number of robots $Dn^{1+\epsilon} < n^{2+\epsilon}$ to explore graphs in an asymptotically optimal number of steps $O(D)$. 
Using the Rotor-Router algorithm allowing only $\log \Delta$ bits per node, an oblivious robot (i.e., robot is not allowed any memory) that also has no knowledge of the entry port when it enters a node, can explore an anonymous port-labeled graph in $2mD$ time steps \cite{DBLP:journals/algorithmica/YanovskiWB03,DBLP:conf/wdag/BampasGHIKK09}.
Menc et al. \cite{DBLP:journals/ipl/MencPU17} proved a lower bound of $\Omega(mD)$ on the exploration time steps for the Rotor-Router algorithm.

The dispersion problem is similar to the problem of scattering or uniform deployment of $k$ robots on a $n$ node graph. The scattering problem was examined on rings \cite{DBLP:journals/tcs/ElorB11,DBLP:conf/podc/ShibataMOKM16}, and on grids \cite{DBLP:journals/ijfcs/BarriereFBS11}, under different system assumptions than those that we make for the dispersion problem.

The dispersion problem is also similar to the load balancing problem, wherein a given load has to be (re-)distributed among several processors. In this analogy, the robots are the load, and it is these active loads rather than the passive nodes that make decisions about movements in the graph. Load balancing in graphs has been studied extensively. Load balancing algorithms use either a diffusion-based approach \cite{DBLP:journals/jpdc/Cybenko89,DBLP:journals/mst/MuthukrishnanGS98,DBLP:conf/spaa/SubramanianS94}, which is somewhat similar to our algorithms, or a dimension-exchange approach \cite{DBLP:journals/jpdc/XuL92} wherein a node can balance with either a single neighbor in a round, or concurrently with all its neighbors in a round.

\section{System Model}
\label{model}
We are given an undirected graph $G$ with $n$ nodes, $m$ edges, and diameter $D$. The maximum degree of any node is $\Delta$. The graph is anonymous, i.e., nodes do not have unique identifiers. At any node, its incident edges are uniquely identified by a label in the range $[0,\delta - 1]$, where $\delta$ is the degree of that node. We refer to this label of an edge at a node as the port number at that node. We assume no correlation between the two port numbers of an edge. There is no memory at the nodes.

In our algorithms, we consider both the synchronous model and the asynchronous model. In the synchronous model, there is a global clock that coordinates the processing of the robots in rounds. In any round, a robot stationed at a node does some computation, perhaps after communication with local robots, and then optionally does a move along one of the incident edges to an adjacent node. Multiple robots can move along an edge in a round. However, we assume that each edge is a single-lane edge, in the sense that robots can move along the edge sequentially. As a result, if multiple robots make a move along an edge, they will enter the node in sequential order which can be captured by a real-time synchronized clock.
In the asynchronous model, there is no global mechanism that coordinates the round numbers of the robots. Thus, each robot executes its rounds/iterations at an independent pace. When a robot determines that it will occupy a particular node in the final configuration, it {\em docks} at that node (by entering $state$ = $settled$). 

The $k$ robots are distinguished from each other by a unique $\lceil\log\,k\rceil$-bit label from the range $[1,k]$. The robots are also endowed with a real-time synchronized clock. A robot can only communicate with other robots that are present at the same node as itself. No robot initially has knowledge of the graph or its parameters $n$, $m$, $D$, and $\Delta$. We assume each robot knows $k$, which is upper-bounded by $n$. In our synchronous algorithms ({\em Helping-Sync}, and the synchronous version of algorithm {\em Independent-Async} presented for the asynchronous model), we assume a robot has knowledge of the parameter $m$ if we want to achieve local termination of the code after a robot has docked at a node in the final configuration. 
For the asynchronous algorithms, the main for-loop counting up to $4m-2(n-1)$ could be replaced by a {\tt while-true} loop. This is because even after a robot docks at a node, it needs to communicate its label (and some additional information in Algorithm {\em Helping-Async}) to visiting robots, to enable them to navigate the graph.

When multiple robots at a node contend to dock at that node, they invoke a MUTEX(node) call that guarantees that only one robot succeeds in docking. The MUTEX may be implemented in various ways. For example, the earliest robot (among the contending robots) that arrived at the node can win the MUTEX; if there is a tie in case of multiple robots arriving simultaneously along different ports, then the tie is broken by choosing the robot arriving along the lowest numbered port as the winner. Or, in the synchronous model, the robots can compare their labels and the robot with the smallest label wins the MUTEX. Or the MUTEX can be implemented by a hardware device to which the winner robot physically connects when it docks. 

\noindent{{\bf Problem Description: }} We are given an initial configuration of $k$ robots, where $k \leq n$, distributed arbitrarily at nodes in the graph. The robots need to move around to reach a final configuration in which there is at most one robot at any node in the graph.

\subsection{Bounds and their Analysis}
\label{analysis}
A lower bound of $\Omega(D)$ on the running time was shown in \cite{DBLP:conf/icdcn/AugustineM18}.  (Note that this prior work \cite{DBLP:conf/icdcn/AugustineM18} required $k=n$ whereas we allow $k \leq n$.)
We present a different lower bound.

\begin{theorem}
The dispersion problem on graphs requires $\Omega(k)$ steps as its running time.
\label{th:timebound}
\end{theorem}
\begin{proof}
Consider a line graph and all $k$ robots colocated at one end node in the initial configuration. In order for the robots to dock at distinct nodes, some robot must travel $k-1$ hops.
\qed
\end{proof}

For dispersion on general unrooted graphs, the best running time  in \cite{DBLP:conf/icdcn/AugustineM18} was $O(m)$. We consider designing space efficient algorithms, subject to a $O(m)$ running time. Observe that DFS based algorithms can run in $O(m)$ time.

A lower bound of $\Omega(\log \, n)$ bits on the memory of robots was shown in \cite{DBLP:conf/icdcn/AugustineM18}. 
For our algorithms, we analyze the memory bounds of robots assuming that a $O(m)$ time algorithm, based on DFS, is to be used. There are two challenges:
\begin{enumerate}
    \item To determine whether a node has been visited before. Note that nodes have no memory in our system model. Although there are $n$ nodes, we observe that a node has been visited before if and only if there is a robot docked at the node and there is a record of having encountered that robot before. As there are $k (\leq n)$ robots, it suffices to track whether or not each of the $k$ robots has been encountered before. This imposes a bound of $O(k)$ bits.  
    \item If it is determined that a node has been visited before, backtracking is in order to meet the $O(m)$ time bound. During the backtracking phase, to determine which port to use for backtracking requires identifying the parent node from which that robot first entered a particular node. Such a parent node can be identified by the local port number of the edge leading to the parent node. A port at a node can be encoded in $\log \,\Delta$ bits. Further, we need to track ports at at most $k-1$ nodes because only a node with a docked robot requires other visiting robots to backtrack, and up to $k-1$ nodes may be occupied by docked robots.  This imposes a bound of $O(k \log\Delta)$ bits.
\end{enumerate}
Thus, the overall bound on memory at a robot is $O(k\log\Delta)$ bits.

\section{Dispersion Using Helping in the Synchronous Model}
\label{helpsync}
To achieve dispersion, each robot begins a DFS-variant traversal of the graph, seeking to identify a node where no other robot has docked. If multiple robots arrive at a node at which no other robot is docked in a particular round, they use the MUTEX(node) function, explained in Section~\ref{model}, to uniquely determine which of those robots can dock at the node. The other robots continue their search for a free node. During this search, a robot needs to determine if the node it visits has been visited before by it. (This is needed to determine whether to backtrack to avoid getting caught in cycles, or continue its forward exploration of the graph.) A node has been visited before if and only if the robot docked there has encountered the visiting robot after it docked. A robot that docks at a node helps other robots to determine whether they have visited this node before. A robot that docks initializes and maintains a boolean array $visited[1,k]$. It sets $visited[r]$ to true if and only if it has encountered robot $r$ after docking. It helps a visiting robot $r$ by communicating to it the value $visited[r]$. 

In order for a robot to determine whether to backtrack from a (already visited) node or resume forward exploration, it needs to know the port leading to the DFS-parent node of the current node. It is helped in determining this as follows. A robot that docks initializes and maintains an array $entry\_port[1,k]$. Subsequently, when a robot $r$ first visits the node, determined using $visited[r]=0$ of the docked node, the $entry\_port[r]$ entry of the docked robot is set to the entry port used by the visiting robot. The docked robot also communicates $entry\_port[r]$ (in addition to $visited[r]$) to a visiting robot $r$ to help it determine whether to backtrack further or resume forward exploration.

A robot uses the following variables: 
\begin{itemize}
\item $port\_entered$ and $parent\_ptr$ of type port can take values from $\{-1,0,1,\ldots,\log\delta-1\}$ ($\lceil\log (\Delta + 1)\rceil$ bits each); $port\_entered$ indicates the port through which the robot entered the current node on the latest visit whereas $parent\_ptr$ is used to track the port through which the robot entered the current node on the first visit; 
\item $state$ (2 bits) can take values from $\{$explore, backtrack, and settled$\}$; and 
\item $seen$ (1 bit) is a boolean to track whether the current node has been seen/visited before. \item $round$ is used as a round counter ($\log \, m= O(\log \,n)$ bits).
\end{itemize}
In addition, a robot initializes the following two arrays once it docks at a node and enters state $settled$: 
\begin{itemize}
\item $visited[1,k]$ of type boolean ($k$ bits), and 
\item $entry\_port[1,k]$ of type port ($k \lceil\log (\Delta + 1)\rceil$ bits). 
\end{itemize}
The semantics of these two arrays was explained above.

\begin{algorithm}[h!]
\caption{Helping-Sync, synchronous execution, code at robot $i$} \label{alg:helpsync}
\begin{algorithmic}[1]
\State Initialize: $port\_entered\gets -1; state\gets explore; parent\_ptr\gets -1$; $seen\gets 0$
\For{$round=0,4m-2(n-1)$} 
\If{$state=settled$}
 \ForAll{other robot $j$ on the node}
   \State send $visited[j]$ and $entry\_port[j]$ to $j$
   \If{$visited[j]=0$}
    \State $visited[j]\gets 1$; $entry\_port[j]\gets$ receive $port\_entered$ from $j$
   \EndIf
 \EndFor
\Else 
 \If{$round > 0$}
 \State $port\_entered, parent\_ptr\gets$ entry port; $seen\gets 0$
 \EndIf
\EndIf
\If{$state=explore$}
 \If{node has a robot $j$ docked in an earlier round}
  \State $seen, parent\_ptr\gets$ receive $visited[i], entry\_port[i]$ from $j$
  \If{$seen=0$}
   \State $parent\_ptr\gets port\_entered$; send $port\_entered$ to $j$ 
  \EndIf
   \If{$seen=1$}
    \State $state\gets backtrack$; move through $port\_entered$
   \EndIf
 \Else
   \If{$i = (r\gets) winner(MUTEX(node))$} 
    \State $i$ docks at node; $state\gets settled$
    \State Initialize $visited[1,k]\gets \overline{0}$; $entry\_port[1,k]\gets \overline{-1}$ 
    \ForAll{ robot $j$ on the node}
     \State $entry\_port[j]\gets$ receive $port\_entered$ from $j$
     \State $visited[j]\gets 1$
    \EndFor
   \Else
    \State send $port\_entered$ to $r$
   \EndIf
  \EndIf
   \If{$state=explore$}
    \State $port\_entered\gets (port\_entered + 1) \mbox{ mod degree of node}$
    \If{$port\_entered=parent\_ptr$}       
     \State $state\gets backtrack$
    \EndIf
    \State move through $port\_entered$
   \EndIf
\ElsIf{$state=backtrack$}
  \State $seen, parent\_ptr\gets$ receive $visited[i], entry\_port[i]$ from docked robot $j$
  \State $port\_entered\gets (port\_entered + 1) \mbox{ mod degree of node}$
   \If{$port\_entered\neq parent\_ptr$}     
    \State $state\gets explore$  
   \EndIf
  \State move through $port\_entered$
\EndIf
\EndFor
\end{algorithmic}
\end{algorithm}

In Algorithm~\ref{alg:helpsync}, lines (3-7): a docked robot $i$ helps visiting robot $j$ by sending it $visited[j]$ and $entry\_port[j]$, and updating the locally maintained $visited[j]$ and $entry\_port[j]$ if this is the first visit of the robot $j$. 

When robot $i$ visits a node where some robot $j$ is already docked, it receives $visited[i]$ and $entry\_port[i]$ from $j$ (line 13). 
If $i$ has $state=explore$ and the node is already visited, $i$ backtracks through $port\_entered$ (lines 16, 17). Whereas if the node is not already visited (lines 14, 15), $i$ sends $port\_entered$ to $j$ which records it in $entry\_port[i]$ (line 7). Robot $i$ contends for the MUTEX (line 19) if there is no robot docked at the node. If $i$ wins the MUTEX and docks, it initializes the data structures $visited[1,k]$ and $port\_entered[i,k]$ and for other robots $j$ concurrently at this node in this round, it fills in their entries in the newly created data structures (lines 19-24). Whereas if $i$ loses the MUTEX contention, it sends $port\_entered$ to the winner of MUTEX (lines 25, 26).
If $i$ has not backtracked and not docked, $state=explore$. In this case (line 27), $i$ increases $port\_entered$ in a modulo fashion (mod degree of node) and moves forward to the next node, but switches $state$ to $backtrack$ if the port to move forward (new value of $port\_entered$) is the same as the entry port (in line 15, $parent\_ptr$ was set to the old value of $port\_entered$, which was set to $entry\_port$ in line 10) (lines 28-31). 

If $i$ has $state=backtrack$ when it visits a node (line 32), it implies some robot $j$ is already docked, and $i$ receives $visited[i]$ and $entry\_port[i]$ from $j$ (line 33). Robot $i$ increases $port\_entered$ in a modulo fashion (mod degree of node) and moves forwards to the next node while switching $state$ to $explore$, unless the port to move along (new value of $port\_entered$) is the parent pointer port (set to $entry\_port[i]$), in which case $i$ keeps $state$ as $backtrack$ and backtracks instead of moving forward (lines 34-37). 


\begin{theorem}
\label{th:helpsync}
Algorithm~\ref{alg:helpsync} (Helping-Sync) achieves dispersion in a synchronous system in $O(m)$ rounds with $O(k\log\Delta)$ bits at each robot.
\end{theorem}
\begin{proof}
Observe that each robot executes a variant of a DFS
in the search for a free node. Each robot may need to traverse each edge of the DFS tree two times (once forward, once backward), and each non-tree edge four times (once for exploration in each direction, and once for backtracking in each direction). So for a total of $4(m-(n-1))+2(n-1)=4m-2n+2$ times. The robot executes for these many rounds, so the running time is $O(m)$. 

From the description and analysis of the variables above, it follows that the memory of each robot is bounded by $O(k\log\Delta)$ bits.

To show that dispersion is achieved in $4m-2n+2$ rounds, observe that the $k$ robots do a collective search of the graph, using individual DFS variants. Within $4m-2n+2$ rounds, if a robot is not yet docked, it will visit each node at least once, and since $k \leq n$, each robot will find a free node and dock there.
\qed
\end{proof}

Note that although a robot may dock at a node, it needs to be active for the rest of the $4m-2n+2$ rounds of the algorithm in order to help other robots which might visit this node.

\section{Dispersion Using Helping in the Asynchronous Model}
\label{helpasync}
Algorithm Helping-Async (Algorithm~\ref{alg:helpasync}) adapts Algorithm Helping-Sync to an asynchronous system but uses the same variables. When a robot arrives at a node, either another robot is docked or not docked at that node; in the latter case, if multiple robots arrive at about the same time, then function MUTEX(node) selects one of them to dock. Another implication of an asynchronous system is that a docked robot needs to loop forever, waiting to help any other robot that might arrive at the node later. 

\begin{theorem}
\label{th:helpasync}
Algorithm~\ref{alg:helpasync} (Helping-Async) achieves dispersion (without termination) in an asynchronous system in $O(m)$ steps with $O(k\log\Delta)$ bits at each robot.
\end{theorem}
\begin{proof}
The proof is similar to that of Theorem~\ref{th:helpsync}. The difference is that due to the nature of the asynchronous system, a docked robot needs to loop forever, waiting to help any other robot that might arrive at the node later. Thus, termination is not possible.
\qed
\end{proof}

As noted in Section~\ref{results}, although the asynchronous algorithm, technically speaking, does not terminate because the docked robots need to be awake to relay local information to visiting robots, we state its time complexity (to be $O(m)$)  because at most $O(m)$ steps are required for each robot to perform active computations and movements until it docks at a node; after that, a docked robot merely passively helps visiting robots (until they find a node to dock).

\begin{algorithm}[h!]
\caption{Helping-Async, asynchronous execution, code at robot $i$} \label{alg:helpasync}
\begin{algorithmic}[1]
\State Initialize: $port\_entered\gets -1; state\gets explore; parent\_ptr\gets -1$; $seen\gets 0$
\For{$round=0,4m-2(n-1)$} 
\If{$round > 0$}
 \State $port\_entered, parent\_ptr\gets$ entry port; $seen\gets 0$
\EndIf
\If{$state=explore$}
 \If{node has a robot $j$ docked}
  \State $seen, parent\_ptr\gets$ receive $visited[i], entry\_port[i]$ from $j$
  \If{$seen=0$}
   \State $parent\_ptr\gets port\_entered$;  send $port\_entered$ to $j$
  \EndIf 
   \If{$seen=1$}
    \State $state\gets backtrack$; move through $port\_entered$
   \EndIf
 \Else
   \If{$i = (r\gets) winner(MUTEX(node))$} 
    \State $i$ docks at node; $state\gets settled$
    \State Initialize $visited[1,k]\gets \overline{0}$; $entry\_port[1,k]\gets \overline{-1}$; break()
   \Else
    \State $seen, parent\_ptr\gets$ receive $visited[i], entry\_port[i]$ from $r$
    \If{$seen=0$}
     \State $parent\_ptr\gets port\_entered$; send $port\_entered$ to $r$ 
    \EndIf
   \EndIf
  \EndIf
    \State $port\_entered\gets (port\_entered + 1) \mbox{ mod degree of node}$
    \If{$port\_entered=parent\_ptr$}        
     \State $state\gets backtrack$
    \EndIf
    \State move through $port\_entered$
\ElsIf{$state=backtrack$}
  \State $seen, parent\_ptr\gets$ receive $visited[i], entry\_port[i]$ from docked robot $j$
  \State $port\_entered\gets (port\_entered + 1) \mbox{ mod degree of node}$
   \If{$port\_entered\neq parent\_ptr$}        
    \State $state\gets explore$  
   \EndIf
  \State move through $port\_entered$
\EndIf
\EndFor
\Repeat \Comment{$state=settled$}
 \ForAll{other robot $j$ that is/arrives at the node}
   \State send $visited[j]$ and $entry\_port[j]$ to $j$
   \If{$visited[j]=0$}
    \State $visited[j]\gets 1$; $entry\_port[j]\gets$ receive $port\_entered$ from $j$
   \EndIf
 \EndFor
\Until{true}
\end{algorithmic}
\end{algorithm}

\begin{algorithm}[h!]
\caption{Independent-Async, asynchronous execution, code at robot $i$} \label{alg:indepasync}
\begin{algorithmic}[1]
\State Initialize: $port\_entered\gets -1; state\gets explore$; $visited[1,k]\gets \overline{0}$; $stack\gets \perp$
\For{$round=0,4m-2(n-1)$} 
\If{$round > 0$}
 \State $port\_entered\gets$ entry port
\EndIf
\If{$state=explore$}
 \If{robot $j$ is docked at node AND $visited[j]=1$}
  \State $state\gets backtrack$; move through $port\_entered$
 \ElsIf{robot $j$ is docked at node AND $visited[j]=0$}
  \State $visited[j]\gets 1$
  \State $push(stack,port\_entered)$
  \State $port\_entered\gets (port\_entered + 1)$ mod degree of node
  \If{$port\_entered=top(stack)$}
   \State $state\gets backtrack; pop(stack)$
  \EndIf
  \State move through $port\_entered$
 \ElsIf{node is free}
  \If{$i = (r\gets) winner(MUTEX(node))$}
   \State $i$ docks at node; $state\gets settled$; break()
  \Else 
   \State $visited[r]\gets 1$
   \State $push(stack,port\_entered)$
   \State $port\_entered\gets (port\_entered + 1)$ mod degree of node
   \If{$port\_entered=top(stack)$}         
    \State $state\gets backtrack; pop(stack)$
   \EndIf
   \State move through $port\_entered$
  \EndIf
 \EndIf  
\ElsIf{$state=backtrack$} 
 \State $port\_entered\gets (port\_entered + 1)$ mod degree of node
 \If{$port\_entered\neq top(stack)$}            
  \State $state\gets explore$ 
 \Else
  \State $pop(stack)$
 \EndIf
 \State move through $port\_entered$
 \EndIf
\EndFor
\end{algorithmic}
\end{algorithm}

\section{Independent Dispersion in the Asynchronous Model}
\label{indepasync}
In Algorithm~\ref{alg:indepasync} (Independent-Async) for the asynchronous model, the traversal of the graph by each robot is the same as in the previous two algorithms. However, there is no helping of undocked robots by docked robots. In addition to $port\_entered$ and $state$, an undocked robot maintains the data following additional data structures:
\begin{itemize}
\item  array of boolean $visited[1,k]$ to determine by checking $visited[r]$ whether it has visited the node where robot $r$ is docked, and 
\item  $stack$ of type port number, to determine the parent pointer of the nodes it has visited. Specifically, the port numbers in the stack (from top to bottom) help the robot to backtrack from the current node all the way to its origin node in the initial configuration. When a robot explores the graph in a step, the entry port number into the current node get pushed onto the stack, and as a robot backtracks in a step, the port number gets popped from the stack. In addition, the top of the stack entry is used for determining whether a robot should switch from backtracking state to explore state, or switch from explore state to backtracking state.
\end{itemize}

Thus, undocked robots are largely independent of docked robots. However, even in this algorithm, a docked robot cannot terminate; it needs to stay up so that it can relay its label $r$ to a visiting undocked robot, which can then look up $visited[r]$, and if necessary, manipulate its $stack$, in order to take further actions for exploring the graph. This action of docked robots (once they enter $settled$ state) is not explicitly shown in the Algorithm~\ref{alg:indepasync} pseudo-code. 

In addition to the $port\_entered$ ($\lceil\log(\Delta + 1)\rceil$ bits) and $state$ (two bits) variables used by the previous algorithms, the boolean $visited[1,k]$ array takes $O(k)$ bits and the $stack$ takes $O(k\log\Delta)$ bits, because the maximum depth of the stack is $k-1$, the maximum number of nodes at which there is a docked robot encountered.

In Algorithm~\ref{alg:indepasync}, when robot $i$ visits a node and $state=explore$ (line 5): 
\begin{enumerate}
\item (lines 6, 7): if a robot is docked and the node has been visited before, robot $i$ backtracks.
\item (lines 8-14): if robot $j$ is docked at the node but the node has not been visited before, robot $i$ marks $visited[j]$ as true and increments $port\_entered$ in a modulo fashion (mod degree of node). If the new value of  $port\_entered$ equals its old value, $i$ changes $state$ to $backtrack$ and moves through $port\_entered$; else the old value of $port\_entered$ is pushed onto the $stack$ and $i$ moves through $port\_entered$ to continue the forward exploration of the graph.
\item (lines 15-24): if the node is free, $i$ contends for the MUTEX to dock. If $i$ wins, it docks. If $i$ loses but robot $r$ won, robot $i$ marks $visited[r]$ as true and increments $port\_entered$ in a modulo fashion (mod degree of node). If the new value of  $port\_entered$ equals its old value, $i$ changes $state$ to $backtrack$ and moves through $port\_entered$; else the old value of $port\_entered$ is pushed onto the $stack$ and $i$ moves through $port\_entered$ to continue the forward exploration of the graph.
\end{enumerate}
When robot $i$ visits a node and $state=backtrack$ (line 25),  robot $i$ increments $port\_entered$ in a modulo fashion (mod degree of node) and moves forward to the next node while switching $state$ to $explore$, unless the port it is going to move along is the parent pointer port (the top of the $stack$), in which case $i$ keeps $state$ as $backtrack$ and pops the top of the $stack$ before moving along (lines 26-31). 

\begin{theorem}
\label{th:indepasync}
Algorithm~\ref{alg:indepasync} (Independent-Async) achieves dispersion (without termination) in an asynchronous system in $O(m)$ steps with $O(k\log\Delta)$ bits at each robot.
\end{theorem}
\begin{proof}
The proof that the running time is $O(m)$, or more specifically $4m-2n+2$ steps, is similar to that of Theorem~\ref{th:helpsync}. 
From the description and analysis of the variables above, it follows that the memory of each robot is bounded by $O(k\log\Delta)$ bits.

Note that due to the nature of the asynchronous system, a docked robot (i.e., once it enters $state$ = $settled$) needs to loop forever, waiting to relay its label to any other robot that might arrive at the node later. (This action is not explicitly shown in Algorithm~\ref{alg:indepasync}.) Thus, termination is not possible. 
\qed
\end{proof}
As noted in Section~\ref{results}, although the asynchronous algorithm, technically speaking, does not terminate because the docked robots need to be awake to relay local information to visiting robots, we state its time complexity (to be $O(m)$)  because at most $O(m)$ steps are required for each robot to perform active computations and movements until it docks at a node; after that, a docked robot merely passively helps visiting robots (until they find a node to dock).

It is a straightforward exercise to transform the algorithm into its synchronous version, {\em Independent-Sync}. In the synchronous algorithm, a robot can terminate after $4m-2(n-1)$ rounds, as it is guaranteed that every other robot would have found a free node by then.

\section{Conclusions}
\label{concl}
For the dispersion problem of mobile robots on general graphs, 
we proposed three algorithms for the synchronous and the asynchronous system models. The algorithms require a memory of $O(k\log\Delta)$ bits at each robot, and a running time of $O(m)$ steps. 
It is a challenge to design more space and time efficient algorithms for dispersion.

\printbibliography

\end{document}